\documentclass[11pt]{article}
\usepackage[letterpaper,margin=1.75in]{geometry}
\usepackage[T1]{fontenc}
\usepackage[utf8]{inputenc}
\usepackage{lmodern}
\usepackage{amsmath,amssymb,amsthm,mathtools}
\newcommand{\llbracket}{[\![}
\newcommand{\rrbracket}{]\!]}
\usepackage{microtype}
\usepackage{hyperref}
\usepackage[numbers]{natbib}

\newtheorem{theorem}{Theorem}[section]
\newtheorem{lemma}[theorem]{Lemma}
\newtheorem{proposition}[theorem]{Proposition}
\newtheorem{corollary}[theorem]{Corollary}
\theoremstyle{definition}
\newtheorem{definition}{Definition}[section]
\theoremstyle{remark}
\newtheorem*{remark}{Remark}

\title{Thinking outside the box is useless \\[1ex] \large NFA = FNFA}
\author{Maxence Ponsardin \and Ville Salo}
\date{\today}

\begin{document}
\maketitle
\begin{abstract}
We study picture-walking automata, namely finite-state automata that accept higher-dimensional parallelotope-/tensor-shaped pictures and are allowed to move in all directions depending on their current state and the currently read symbol. It is a long-standing open problem whether the nondeterministic such automata become stronger if automata are allowed to exit the picture. In this paper, we resolve the problem in full generality: NFA = FNFA, for pictures of any dimension.
\end{abstract}

\section{Introduction}

In one-dimensional formal language theory, the class of regular languages is highly robust, and a large family of finite-state models accepts precisely this class of languages. In particular deterministic, nondeterministic, and alternating finite-state automata all define precisely this class. This also remains true even if one allows automata to revisit previous cells, by allowing the automata to choose the direction of movement.

In higher-dimensional formal language theory, one-dimensional words are replaced by matrices or higher-dimensional tensors (or geometrically $d$-dimensional parallelotopes) of symbols. In dimension two and higher, the situation is the opposite, and essentially any two models define a different set of languages. We concentrate here on the models that allow movement in all directions, called \textbf{picture-walking automata}. Concretely, the automaton starts in some fixed spot (like a corner) of the picture in any initial state, follows its transition rule to move around the picture and change states depending on the symbol under the head, and finally may or may not enter an accepting state.

Here, it is in particular known that deterministic finite-state automata (DFA) are strictly weaker than nondeterministic finite-state automata (NFA), and this remains true even if one restricts to unary languages \cite{10.1007} or square-shaped languages \cite{Kari2f011}

An interesting notion due to Rosenfeld \cite{rosenfeld2014picture} from 1979 is to allow such automata to exit the picture frame, where they see only empty space. These are called \emph{frameless} automata. For the DFA model, it is rather trivial that this capability is not useful. For the FNFA (frameless NFA) model, intuition suggests that the analogous result remains true, but this is not as easy to show.

Rosenfeld solved this problem for one-dimensional words (where the word can only be exited left or right), and left the problem open in two dimensions. Kari and Moore proved a technical result in \cite{10.1007}, giving a strong result for two-dimensional automata that are not allowed to \textbf{enter} the picture. This was used by the second-named author \cite{salo2011classes} to obtain NFA = FNFA in dimension two. This method does not work in three or higher dimensions, and this question was left open in \cite{Kari2f011} In fact, it does not even resolve the NFA = FNFA problem for one-dimensional words thought of as three-dimensional pictures.

The contribution of the present paper is to solve the general case (for all parallelotope-shaped pictures in all dimensions).

\begin{theorem}
NFA = FNFA holds in all dimensions
\end{theorem}

The crucial new method is to analyze the shape of paths used to reach one position from another, while staying outside the picture frame. Specifically, we prove the \textbf{Navigation Lemma} (Lemma 1), which states that if the ``minimal'' path from a position-state pair $(q, p)$ to another such pair $(q', p')$ stays sufficiently far outside the picture frame, then this path consists of a bounded number of ``straight'' geometric segments. We call such a path a \emph{zig-zag}.

In the case where there is no picture at all, the Navigation Lemma essentially simplifies to the known result that a regular language's Parikh image is realized by a bounded sublanguage \cite{10.1145}, although we apply a linear projection that only keeps track of movement of the automaton, for example a left move is considered the inverse of a right move. Up to this difference, the Navigation Lemma can be seen as a weak analog of the result of \cite{10.1145} with a ``parallelotope-shaped hole'' of forbidden Parikh images.

Once we have proved the Navigation Lemma, the main result is proved roughly as follows. By analyzing the minimal path between two position-state pairs, it is possible to divide this path into subpaths such that some of them stay close enough to the picture to be simulated while the others are far enough to have a zig-zag shape. 

As a zig-zag only moves between the half-spaces outside the picture only a bounded number of times, we can simulate them inside the picture using essentially methods from the two-dimensional theory, namely it is known that when an FNFA exits a picture from one side, its ``configuration of possible landings'' back to the hyperplane extending the side is semi-linear. The two-dimensional version appears in \cite{10.1007}, and is called the \textbf{Landing Lemma} in \cite{salo2011classes}

\section{Definitions}

\begin{definition}
For $N, M \in \mathbb{Z}$, we define by $\llbracket N, M\rrbracket$ the sets of integers whose values are between $N$ and $M$. I.e.\ $\llbracket N, M\rrbracket = \{x \in \mathbb{Z} \;|\; N \le x \le M\}$.
\end{definition}

\begin{definition}
Let $\Sigma$ be an alphabet, a picture $b$ over $\Sigma$ is the $d$-dimensional analog of a word. I.e.\ a picture is a $d$-dimensional tensor such that its coefficients are in $\Sigma$. We denote $\Sigma_*^*$ the set of pictures over $\Sigma$. For a vector $p \in \mathbb{Z}^d$, we denote by $b_{p}$ the content of the cell\textbf{ in position }$p$\textbf{ of }$b$. If $p$ is not a cell of $b$ then $b_{p}$ is set to be $\#$, a special character that is not in the alphabet.
We write $\operatorname{dom}(b)$ for the set of positions of $(\mathrm{non\text{-}}\#)$ cells of $b$, called the domain of $b$ and $\operatorname{edom}(b)$ for the cells of $\operatorname{dom}(b)$ and their immediate cardinal and diagonal neighbors.
\end{definition}

\begin{definition}
Let $b \in \Sigma_*^*$ be a picture, the dimension of $b$ is the unique tuple $(d_{1}, \dots, d_{d}) \in (\mathbb{N}^*)^d$ such that $\operatorname{dom}(b) = \{(a_{1}, \dots, a_{d}) \in \mathbb{Z}^d | \forall i \in \llbracket 1, d\rrbracket, 1 \le a_{i} \le d_{i}\}$. The distance $d$ between two vectors in $\mathbb{Z}^d$ is given by default in the infinity norm $\ell^\infty$ of their distance i.e.\ if we have $v_{1} = (a_{1}, \dots, a_{d}) \in \mathbb{Z}^d$ and $v_{2} = (b_{1}, \dots, b_{d}) \in \mathbb{Z}^d$, the distance between $v_{1}$ and $v_{2}$ is $d(v_{1}, v_{2}) = \max(|a_{i} -b_{i}|)$. We say that a vector $v = (a_{1}, \dots, a_{d})$ is at a distance $L$ from the picture iff $L = \min(\{d(v, u) | u \in \operatorname{dom}(b)\})$. Finally, a vector $v = (a_{1}, \dots, a_{d}) \in \mathbb{Z}^d$ is at a distance at least $L$ from the picture iff the distance between $v$ and the picture is at least $L$ i.e.\ $\exists i \in \llbracket 1, d\rrbracket, a_{i} < 1 - L \lor a_{i} > d_{i} + L$.
\end{definition}

\begin{definition}
A picture walking Nondeterministic Finite state Automata (NFA) is a 6-tuple $(\Sigma, Q, M, I, F, \delta)$ such that:
\begin{itemize}
\item $\Sigma$ is a finite alphabet. Usually $\Sigma = \{0, 1\}$
\item $Q$ is a finite set of states
\item $M \in \mathbb{Z}^d$ is a finite set of moves. Usually $M$ is the sets of vectors of $\mathbb{Z}^d$ that have $\ell^1$-norm at most 1. We write the zero movement vector as $0 \in M$.
\item $\delta: Q \times (\Sigma \cup \{\#\}) \to 2^{Q \times M}$ is a transition function
\item $I \subset Q$ is a set of initial states.
\item $F \subset Q$ is a set of final states
\end{itemize}
The Frameless Nondeterministic Finite state Automata (FNFA) are defined by the same data.
\end{definition}

The automata are defined by the same data, but their acceptance conditions are slightly different:

\begin{definition}
Let $b \in \Sigma_*^*$ be a picture and $A = (\Sigma, Q, M, I, F, \delta)$ be a picture walking \textbf{NFA}. We say that $A$ accepts the picture $b$ iff there is a finite sequence of states $q_{0}, \dots, q_{N}$ and a finite sequence of positions $p_0, \dots, p_N$ such that:
\begin{itemize}
\item $q_{0} \in I$ and $p_{0} = (1, \dots, 1)$
\item $q_{N} \in F$
\item \textbf{$\forall i \in \llbracket 0, N\rrbracket, p_{i} \in \operatorname{edom}(b)$} \hfill (*)
\item $\forall i \in \llbracket 0, N-1\rrbracket, (q_{i+1}, p_{i+1} - p_{i}) \in \delta(q_{i}, b_{p_i})$
\end{itemize}
We call $L(A)$ the set of pictures that are accepted by $A$. For \textbf{FNFA}, we use the exact same definition, except the third item (marked with (*)) is removed.
\end{definition}

\begin{remark}
The only difference between a NFA and a FNFA is that during the acceptance of a picture, the FNFA is allowed to go outside $\operatorname{edom}(b)$. Alternatively, we could simply define NFA as those FNFA which never \textbf{attempt} to exit the immediate neighborhood of the picture (if started normally from a corner, in an initial state), and then use the same acceptance condition for both. This condition is actually undecidable, but in any case the class of languages defined by NFA would remain the same with this alternative definition.
\end{remark}

\begin{definition}
Let $A = (\Sigma, Q, M, I, F, \delta)$ be a picture walking FNFA on $\mathbb{Z}^d$. A \emph{configuration} is a pair $(q, p)$ where $q \in Q$ and $p \in \mathbb{Z}^d$.
\end{definition}

\begin{definition}
Let $A = (\Sigma, Q, M, I, F, \delta)$ be a picture walking FNFA on $\mathbb{Z}^d$ and $b \in \Sigma_*^*$ be a picture. A path between two configurations $(q_{s}, p_{s})$ and $(q_{t}, p_{t})$ is a finite sequence of configurations $(q_{0}, p_{0}), \dots, (q_{N}, p_{N}) \in Q \times \mathbb{Z}^d$ such that:
\begin{itemize}
\item $(q_{s}, p_{s}) = (q_{0}, p_{0})$
\item $(q_{t}, p_{t}) = (q_{N}, p_{N})$
\item $\forall i \in \llbracket 0, N-1\rrbracket, (q_{i+1}, p_{i+1}-p_{i}) \in \delta(q_{i}, b_{p_i})$ where $b_{p_i}$ is the letter read at position $p_{i}$ on the picture $b$.
\end{itemize}
A path can be represented by the starting position $p_{0}$ and the following word:
$q_{s} = q_{0} \overset{p_{1} - p_{0}}\longrightarrow q_{1} \overset{p_{2} - p_{1}}\longrightarrow \dots \overset{p_{N} - p_{N-1}}\longrightarrow q_{N}  = q_{t}$
\end{definition}

\begin{definition}
Let $A = (\Sigma, Q, M, I, F, \delta)$ be a picture walking FNFA on $\mathbb{Z}^d$ and $b \in \Sigma_*^*$ be a picture. A word of the automaton $A$ is a sequence of states and direction such that there exists a path between two configurations that has exactly this sequence of states and directions.
\end{definition}

\begin{definition}
Let $A = (\Sigma, Q, M, I, F, \delta)$ be a picture walking FNFA on $\mathbb{Z}^d$ and $b \in \Sigma_*^*$ be a picture. A path between two configurations $(q_{s}, p_{s})$ and $(q_{t}, p_{t})$ that avoids the picture is a path $(q_{0}, p_{0}), \dots, (q_{N}, p_{N})$ such that $\forall i \in \llbracket 1, N\rrbracket, p_{i} \notin \operatorname{dom}(b)$.
\end{definition}

\begin{definition}
Let $A = (\Sigma, Q, M, I, F, \delta)$ be a picture walking FNFA on $\mathbb{Z}^d$ and $b \in \Sigma_*^*$ be a picture. The minimal path between two configurations $(q_{s}, p_{s})$ and $(q_{t}, p_{t})$ is the finite sequence $(q_{0}, p_{0}), \dots, (q_{N}, p_{N}) \in Q \times \mathbb{Z}^d$ of configurations such that:
\begin{itemize}
\item $(q_{0}, p_{0}), \dots, (q_{N}, p_{N})$ is a path between $(q_{s}, p_{s})$ and $(q_{t}, p_{t})$
\item For all other paths $(q'_{0}, p'_{0}), \dots, (q'_{N'}, p'_{N'})$, we either have $N < N'$ or we have $N = N' \land (q_{0}, p_{0}, \dots, q_{N}, p_{N}) \le_{\mathrm{lex}} (q'_{0}, p'_{0}, \dots, q'_{N}, p'_{N})$
\end{itemize}
\end{definition}

\begin{remark}
Similarly, we can define the minimal path between two configurations that avoids the picture.
\end{remark}

\section{Results}

\begin{lemma}
\textbf{\emph{(The Navigation Lemma)}}. Let $A = (\Sigma, Q, M, I, F, \delta)$ be a picture walking FNFA on $\mathbb{Z}^d$ and $b \in \Sigma_*^*$ be a picture. There exists $M_{1}$ such that for every pair of configurations, if the minimal path that avoids the picture stays at a distance at least $|Q|$ from the picture, then its word is of the shape $u_{1} v_{1}^{l_{1}} u_{2} \dots u_{M_1} v_{M_{1}}^{l_{M_1}} u_{M_1+1}$ with $u_{1}, \dots, u_{M_{1}+1}, v_{1}, \dots, v_{M_1}$ words of size at most $|Q|$ and $l_{1}, \dots, l_{M_1} \in \mathbb{N}$.
\end{lemma}

To show this result, we need to introduce new definitions. In the following, recall that we can represent paths with words, and we mostly identify the two.

\begin{definition}
Let $A = (\Sigma, Q, M, I, F, \delta)$ be a picture walking FNFA on $\mathbb{Z}^d$ and $b \in \Sigma_*^*$ be a picture. Let W be the set of words of the automaton $A$. We define these two sets:
$U = \{q_0 \xrightarrow{a_1} q_1 \cdots \xrightarrow{a_m} q_m \in P \mid \forall i,j \in \llbracket 0,m\rrbracket,\ i\ne j \Rightarrow q_i\ne q_j\}$ and $V = \{q_0 \xrightarrow{a_1} q_1 \cdots \xrightarrow{a_m} q_m \in P \mid q_0=q_m,\ \forall i,j \in \llbracket 0,m-1\rrbracket,\ i\ne j \Rightarrow q_i\ne q_j\}$. 
$U$ is the set of words that do not contain any repetition of states and $V$ is the set of words that start and end in the same state but there is no other repetition. A \emph{decomposition} of a word $w=q_0 \xrightarrow{a_1} q_1 \cdots \xrightarrow{a_N}q_N$ is a tuple $(u_1,v_1,l_1,\dots,u_K,v_K,l_K)$ such that $w=u_1v_1^{l_1}u_2\cdots u_Kv_K^{l_K}$ with $u_1,\dots,u_K\in U$, $v_1,\dots,v_K\in V$, and $l_1,\dots,l_K\in\mathbb N$. The \emph{minimal decomposition} of $w$ is the lexicographically least such decomposition among those with minimal $K$.
\end{definition}

\begin{remark}
The decomposition of a word is not necessarily unique. However, every path has a decomposition. For a word $w=q_0 \xrightarrow{a_1} q_1 \cdots \xrightarrow{a_N}q_N$, one elementary decomposition is given below.
\[
\forall i \in \llbracket 1, N\rrbracket, (u_{i}, v_{i}, l_{i}) \quad = \begin{cases}
(q_{i-1} \overset{a_{i}}\longrightarrow q_{i}, q_{i}, 0) &\mathrm{ if } q_{i-1} \ne q_{i}\\
(q_{i-1}, q_{i-1} \overset{a_{i}}\longrightarrow q_{i}, 1) & \mathrm{otherwise}\end{cases} \]
\end{remark}

\begin{proof}[Proof of Lemma 1]
Let $A = (\Sigma, Q, M, I, F, \delta)$ be a picture walking FNFA on $\mathbb{Z}^d$ and $b \in \Sigma_*^*$ be a picture. Let $(q_{s}, p_{s})$, $(q_{t}, p_{t})$ be two configurations and $(q_{0}, p_{0}), \dots, (q_{N}, p_{N})$ the minimal path between these two configurations that avoids the picture and we denote its word $w$. Note $k = |Q|$ and suppose that $w$ stays at a distance at least $k$ from the picture. We consider the minimal decomposition $w = (u_{1}, v_{1}, l_{1}, \dots, u_{K}, v_{K}, l_{K})$.

Firstly, we can claim that: $\forall i \in \llbracket 1, K-1\rrbracket, l_{i} \ne 0$. Suppose there exists $i \in \llbracket 1, K-1\rrbracket$ such that $l_{i} = 0$, we have that
\[ (u_{1}, v_{1}, l_{1}, \dots, u_{i}, v_{i}, 0, u_{i+1}, v_{i+1}, l_{i+1}, \dots, u_{K}, v_{K},l_{K}) \]
is the minimal decomposition of $w$.

Now we have two cases:
\begin{itemize}
\item Case 1: $u_{i} u_{i+1} \in U$. Then $(u_{1}, v_{1}, l_{1}, \dots, u_{i} u_{i+1}, v_{i+1}, l_{i+1}, \dots, u_{K}, v_{K}, l_{K})$  is a valid decomposition of $w$. This is a contradiction, because it has less elements than $(u_{1}, v_{1}, l_{1}, \dots, u_{i}, v_{i}, 0, u_{i+1}, v_{i+1}, l_{i+1}, \dots, u_{K}, v_{K},l_{K})$.
\item Case 2: $u_{i} u_{i+1}$ has a state that appears twice in the word (the first time in $u_{i}$ and the second time in $u_{i+1}$. Therefore, we can write $u_{i} u_{i+1}$ as $u'_{i} v'_{i} u'_{i+1}$ with $u'_{i}, u'_{i+1} \in U$, $v'_{i} \in V$ and $u'_{i}$ is a strict prefix of $u_{i}$. So $(u_{1}, v_{1}, l_{1}, \dots, u'_{i}, v'_{i}, 1, u'_{i+1}, v_{i+1}, l_{i+1}, \dots, u_{K}, v_{K}, l_{K})$ is a valid decomposition of $w$. This is a contradiction, because $u'_{i} < u_{i}$.
\end{itemize}
Therefore, $\forall i \in \llbracket 1, K-1\rrbracket, l_{i} \ne 0$

Secondly, we can claim that: $\forall i, j \in \llbracket 1, K-1\rrbracket, i \ne j \Rightarrow v_{i} \ne v_{j}$. Suppose there exist $i, j \in \llbracket 1, K-1\rrbracket$ such that $i \ne j$ and $v_{i} = v_{j}$. Therefore, these two following tuples are decompositions of words of paths between $(q_{s}, p_{s})$ and $(q_{t}, p_{t})$ that avoid the picture:
\begin{itemize}
\item $(u_{1}, v_{1}, l_{1}, \dots, u_{i}, v_{i}, l_{i} - 1, \dots, u_{j}, v_{j}, l_{j} + 1, \dots, u_{K}, v_{K}, l_{K})$
\item $(u_{1}, v_{1}, l_{1}, \dots, u_{i}, v_{i}, l_{i} + 1, \dots, u_{j}, v_{j}, l_{j} - 1, \dots, u_{K}, v_{K}, l_{K})$
\end{itemize}

They avoid the picture because at each time, they stay at a distance at most $|v_{i}|$ from $w$, but $|v_{i}| \le k$ and $w$ stay at a distance at least $k$ from the picture. By minimality of $w$ we have these inequalities:
\begin{align*}
u_1v_1^{l_1}\cdots u_iv_i^{l_i}\cdots u_jv_j^{l_j}\cdots u_Kv_K^{l_K} &\le_{\mathrm{lex}} u_1v_1^{l_1}\cdots u_iv_i^{l_i-1}\cdots u_jv_j^{l_j+1}\cdots u_Kv_K^{l_K}, \tag{1}\\
u_1v_1^{l_1}\cdots u_iv_i^{l_i}\cdots u_jv_j^{l_j}\cdots u_Kv_K^{l_K} &\le_{\mathrm{lex}} u_1v_1^{l_1}\cdots u_iv_i^{l_i+1}\cdots u_jv_j^{l_j-1}\cdots u_Kv_K^{l_K}. \tag{2}
\end{align*}

We denote by $z$ the word $v_{i}^{l_{i} - 1} u_{i+1} v_{i+1}^{l_{i+1}} \dots u_{j-1} v_{j-1}^{l_{j-1}} u_{j} v_{j}^{l_{j} - 1}$, we have these inequalities:
\begin{align*}
v_i z v_j &\le_{\mathrm{lex}} z v_jv_j, \tag{1}\\
v_i z v_j &\le_{\mathrm{lex}} v_iv_i z. \tag{2}
\end{align*}

As, $v_{i} = v_{j}$, we have $\overbrace{v_{i} z \le_{\mathrm{lex}}}^{(1)} z v_{i}  \overbrace{\le_{\mathrm{lex}} v_{i} z}^{(2)}$. Therefore we have $v_{i} z = z v_{i}$ which means that there exists a word $x$ and two integers $p, q$ such that $v_{i} = x^p$ and $z = x^q$. But, as there is only one repetition of states in $v_{i}$, we have $p = 1$ and $z = v_{i}^q$. Therefore $(u_{1}, v_{1}, l_{1}, \dots, u_{i}, v_{i}, 2 + q, u_{j+1}, v_{j+1}, l_{j+1}, \dots, u_{K}, v_{K}, l_{K})$ is a valid decomposition of $w$. This is a contradiction, because it has fewer elements than $(u_{1}, v_{1}, l_{1}, \dots, u_{K}, v_{K}, l_{K})$.

Therefore, $\forall i, j \in \llbracket 1, K-1\rrbracket, i \ne j \Rightarrow v_{i} \ne v_{j}$.

Finally, we can claim that, $\forall i: |u_{i}| \le k \mathrm{and} |v_{i}| \le k$, as otherwise it would have more than one repetition of states. And we can also claim that $K - 1 \le k^{k+1}(2d)^k$ because all $v_{i}$ for $i \le K-1$ are different, but there exist only $k^{k+1}(2d)^k$ different words in $V$.

To conclude the proof, we just need to set $M_{1} = k^{k+1}(2d)^k + 1$.
\end{proof}

\begin{remark}
Using basic algebra over the vector space $\mathbb{Q}^d$, it is possible to show that we only need $M_{1} = d$, but if we do that we could not guarantee that the path will not hit the picture.
\end{remark}

\begin{definition}
Let $A = (\Sigma, Q, M, I, F, \delta)$ be a picture walking FNFA on $\mathbb{Z}^d$ and $b \in \Sigma_*^*$ be a picture, let $(d_{1}, \dots, d_{d}) \in (\mathbb{N}^*)^d$ be the dimension of $b$. For every $i \in \llbracket 1, \dots, d\rrbracket$, we define the bottom and top half-spaces of the $i^{\mathrm{th}}$ dimension the sets $H_{i}^{(b)} = \{(v_{1}, \dots, v_{d}) \in \mathbb{Z}^d | v_{i} < 1\}$ and $H_{i}^{(t)} = \{(v_{1}, \dots, v_{d}) \in \mathbb{Z}^d \;|\; v_{i} > d_{i}\}$.

Let $(q_{0}, p_{0}), \dots, (q_{N}, p_{N})$ be a path between two configurations $(q_{s}, p_{s})$ and $(q_{t}, p_{t})$ that avoids the picture. A \textbf{crossed half-spaces sequence} of the path is a finite sequence of half-spaces $H_{1}, H_{2}, \dots, H_{K}$ such that there exists an associated increasing sequence of integer $0 = i_{1} < i_{2} < \dots < i_{K} < i_{K+1} = N+1$ with the following property:
\begin{itemize}
\item $\forall j \in \llbracket 1, K\rrbracket, p_{i_j} \in H_{j} \land p_{i_{j}+1} \in H_{j} \land \dots \land p_{i_(j+1)-1} \in H_{j}$
\item $\forall j \in \llbracket 2, K\rrbracket, p_{i_j} \notin H_{j-1}$
\end{itemize}
I.e. a crossed half-spaces sequence of the path is a sequence of half-spaces such that the path starts in the first half-spaces and each time it leaves a half-space of the sequence, the path is in the next one.
\end{definition}

\begin{remark}
As $(\bigcup_{i=1}^d H_{i}^{(t)} \cup H_{i}^{(b)}) = \mathbb{Z}^d \setminus \operatorname{dom}(b)$, a path that avoids the picture is at each step in a half-space, so for every path between two configurations that avoids the picture, there is at least one sequence crossed half-space by the path.
\end{remark}

\begin{definition}
Let $A = (\Sigma, Q, M, I, F, \delta)$ be a picture walking FNFA on $\mathbb{Z}^d$ and $b \in \Sigma_*^*$ be a picture. Let $(q_{0}, p_{0}), \dots, (q_{N}, p_{N})$ be a path between configurations $(q_{s}, p_{s})$ and $(q_{t}, p_{t})$. We say that a crossed half-spaces sequence $H_{1}, H_{2}, \dots, H_{K}$ is \textbf{good} if it has the following property:
\begin{itemize}
\item $H_{1}, H_{2}, \dots, H_{K}$ is indeed a half-space sequence with associated finite integer sequence $0 = i_{1} < i_{2} < \dots < i_{K} < i_{K+1} = N+1$
\item For all $j \in \llbracket 1, K\rrbracket$, let $p_{i_j} = (a_{1}, \dots, a_{d})$, we have:
$(\exists i \in \llbracket 1, d\rrbracket: H_{j} = H_{i}^{(b)} \land a_{i} < 1-|Q|) \lor (\exists i \in \llbracket 1, d\rrbracket : H_{j} = H_{i}^{(t)} \land a_{i} > d_{i} + |Q|)$
\end{itemize}
I.e.\ a crossed half-space sequence $H_{1}, H_{2}, \dots, H_{K}$ of the path is \textbf{good} if each time we change half-space, we are $|Q|$ deep inside the half-space.
\end{definition}

\begin{remark}
Let $(q_{0}, p_{0}), \dots, (q_{N}, p_{N})$ be a path that stays at a distance at least $k$ from the picture. This path has a good crossed half-spaces sequence. At each step, it is a distance at least $|Q|$ from the picture so there is at least one dimension in which the position is $k$ deep in a half-space.
\end{remark}

\begin{lemma}
Let $A = (\Sigma, Q, M, I, F, \delta)$ be a picture walking FNFA on $\mathbb{Z}^d$ and $b \in \Sigma_*^*$ be a picture. There exists $M_{2} \in \mathbb{N}$ such that for every pair of configurations, if the minimal path that avoids the picture stays at a distance at least $|Q|$ from the picture, then each of its good half-spaces sequences have length at most $M_{2}$.
\end{lemma}

\begin{proof}
Let $A = (\Sigma, Q, M, I, F, \delta)$ be a picture walking FNFA on $\mathbb{Z}^d$ and $b \in \Sigma_*^*$ be a picture. Let $(q_{s}, p_{s})$, $(q_{t}, p_{t})$ be two configurations and $(q_{0}, p_{0}), \dots, (q_{N}, p_{N})$ the minimal path between these two configurations that avoids the picture and suppose it stays at a distance at most $|Q|$ from the picture. Let $H_{1}, H_{2}, \dots, H_{K}$ be a good crossed half-space sequence of the path.

We denote $w$ the word, by applying the Lemma 1, there exists $M_{1} \in \mathbb{N}$ that depends only on the automaton, $u_{1}, \dots, u_{M_1}, v_{1}, \dots, v_{M_1}$ words of size at most $|Q|$ and $l_{1}, \dots, l_{M_1} \in \mathbb{N}$ such that $w = u_{1} v_{1}^{l_{1}} \dots u_{M_1} v_{M_{1}}^{l_{M_1}}$.

Let $i \in \llbracket 1, M_{1}\rrbracket$, suppose at the beginning of $u_{i}$ we are in half-space $H_{j}$ for some $j \in \llbracket 1, K\rrbracket$. Because the half-spaces sequence is good, if $u_{i}$ leaves $H_{j}$ and is now in $H_{j+1}$, it will need more than $|Q|$ steps to leave $H_{j+1}$, but $u_{i}$ has length less than $|Q|$. Therefore $u_{i}$ can leave at most one half-space of the sequence.

Let $i \in \llbracket 1, M_{1}\rrbracket$, suppose at the beginning of $v_{i}^{l_i}$ we are in half-space $H_{j_1}$ for some $j_{1} \in \llbracket 1, K\rrbracket$. There exists $j_{2} \in \llbracket 1, K\rrbracket$ such that $H_{j_1}, \dots, H_{j_2}$ is a half-space sequence for the subpath $v_{i}^{l_i}$.

Suppose there are $l_{1}, l_{2} \in \llbracket j_{1} + 1, j_{2}\rrbracket$ such that $l_{1} \ne l_{2}$ and $H_{l_1} = H_{l_2} = H_{r}^{(b)}$ for some $r \in \llbracket 1, d\rrbracket$. Let $p_{l_1} = (a_{1}^{(1)}, \dots, a_{d}^{(1)})$ and $p_{l_2} = (a_{1}^{(2)}, \dots, a_{2}^{(2)})$ be the position in which we are for the first time in $H_{l_1}$ and $H_{l_2}$. Because the overall crossed half-space sequence is good, we have $a_{r}^{(1)} < 1 - |Q|$ and $a_{r}^{(2)} < 1 - |Q|$. But as we have left $H_{l_1}$, there exists $l_{3} \in \llbracket l_{1}, l_{2}\rrbracket$ such that $p_{l_3} = (a_{1}^{(3)}, \dots, a_{d}^{(3)})$ with $a_{r}^{(3)} \ge 1$. This means that between $p_{l_1}$ and $p_{l_3}$, we have a sequence of $|v_{i}| \le |Q|$ moves that goes positively in the $r^{\mathrm{th}}$ and between $p_{l_3}$ and $p_{l_2}$, we have a sequence of $|v_{i}| \le |Q|$ moves that goes negatively in the $r^{\mathrm{th}}$ dimension. This is a contradiction, because sequence of $|v_{i}|$ moves are exactly $v_{i}$ in a different order. Therefore we cannot have $l_{1}, l_{2} \in \llbracket j_{1} + 1, j_{2}\rrbracket$ such that $l_{1} \ne l_{2}$ and $H_{l_1} = H_{l_2} = H_{r}^{(b)}$. Similarly, we cannot have $l_{1}, l_{2} \in \llbracket j_{1} + 1, j_{2}\rrbracket$ such that $l_{1} \ne l_{2}$ and $H_{l_1} = H_{l_2} = H_{r}^{(t)}$.

Therefore $v_{i}^{l_i}$ can leave at most $2d$ half-spaces of the sequence.

Finally, the total path $u_{1} v_{1}^{l_{1}} \dots u_{M_{1}}^{l_{M_1}}$ can leave at most $M_{1}(2d+1)$ half-spaces so $K \le M_{1}(2d+1) + 1$.
To conclude the proof, we just need to set $M_{2} = M_{1}(2d+1)+1$.
\end{proof}

\begin{definition}
Let $A = (\Sigma, Q, M, I, F, \delta)$ be a picture walking FNFA on $\mathbb{Z}^d$ and $b \in \Sigma_*^*$ be a picture. Let $A' = (\Sigma, Q', M, I', F', \delta')$ be a picture walking NFA and a function $f: Q' \times \mathbb{Z}^d \to Q \times \mathbb{Z}^d$. Let $(p_{s}, q_{s})$ and $(p_{t}, q_{t})$ be two configurations of $A$ such that there exists a path between them. We say that $(A', f)$ can simulate the path from $(q_{s}, p_{s})$ to $(q_{t}, p_{t})$ if there exists a finite sequence of configurations of $A'$ $(q'_{0}, p'_{0}), \dots, (q'_{N}, p'_{N})$ such that:
\begin{itemize}
\item $f(q'_{0}, p'_{0}) = (q_{s}, p_{s})$
\item $f(q'_{N}, p'_{N}) = (q_{t}, p_{t})$
\item $\forall i \in \llbracket 0, N-1\rrbracket, (q'_{i+1}, p'_{i+1}-p'_{i}) \in \delta'(q'_{i}, b_{p'_i})$ where $b_{p'_i}$ is the letter read at position $p'_{i}$ on the picture $b$.
\end{itemize}
\end{definition}

\begin{proposition}
\label{prop:SimulateClose}
Let $A = (\Sigma, Q, M, I, F, \delta)$ be a picture walking FNFA on $\mathbb{Z}^d$ and $b \in \Sigma_*^*$ be a picture. There exists a picture walking NFA $A' = (\Sigma, Q', M, I', F', \delta')$ and a function $f: Q' \times \mathbb{Z}^d \to Q \times \mathbb{Z}^d$ that can simulate exactly the paths between two configurations of $A$ that stay at a distance at most $C$ from the picture for any $C \in \mathbb{N}$.
\end{proposition}

\begin{proof}
Let $A = (\Sigma, Q, M, I, F, \delta)$ be a picture walking FNFA on $\mathbb{Z}^d$, $b \in \Sigma_*^*$ be a picture and $C \in \mathbb{N}$. We denote $(d_{1}, \dots, d_{d})$ the dimension of the picture. The idea is to construct $A' = (\Sigma, Q', M, I', F', \delta')$ that simulates $A$ normally but when it leaves the picture, we use new states to remember in which position it is.
The way we do that is we create states that will store how much we are deep in each direction positively or negatively from the picture (as paths are at a distance at most $C$, the number of states will be finite). Now we try to simulate as $A$ would do but we perform each step in two stages:
\begin{itemize}
\item The first one will do the step normally
\item The second one will check if we have left the picture and it will bring us back inside it and change state to remember we moved outside the picture.
\end{itemize}

We give the formal definition of the automaton. Denote $Q'_{\mathrm{in}} = Q \times (\llbracket -C, C\rrbracket)^d$ and $Q'_{\mathrm{check}} = Q \times M \times (\llbracket -C, C\rrbracket)^d$. Let $Q' = Q'_{\mathrm{in}} \cup Q'_{\mathrm{check}}$ and define $f: Q' \times \mathbb{Z}^d \to Q \times \mathbb{Z}^d$ as follows:
\[
f(q,p)=\begin{cases}
(q',p+p') & \text{if } q=(q',p')\in Q'_{\mathrm{in}},\\
(q',p+p') & \text{if } q=(q',m',p')\in Q'_{\mathrm{check}}.
\end{cases}
\]
We define $\delta'$ as follows:
\item $\forall (q, p) \in Q'_{\mathrm{in}}, \forall a \in \Sigma$:
$\delta'((q, p), a) = \bigcup_{(q', m') \in \delta(q, a)} ((q', m', p), m') \subseteq Q'_{\mathrm{check}} \times M$
\item $\forall (q, m, p) \in Q'_{\mathrm{check}}, \forall a \in \Sigma:$
\[
\delta'((q,m,p),a)=\begin{cases}
((q,p+m),-m) & \text{if } a=\#,\\
((q,\widetilde p),0) & \text{otherwise}.
\end{cases}
\]
Above, $\widetilde p$ resets the coordinate moved by $m$.

Let $I' = I \times (\llbracket -C, C\rrbracket)^d \subseteq Q'_{\mathrm{in}}$ and $F' = F \times (\llbracket -C, C\rrbracket)^d \subseteq Q'_{\mathrm{in}}$.

Let $(q_{s}, p_{s})$ and $(q_{t}, p_{t})$ be two configurations of $A$ and $(q_{0}, p_{0}), \dots, (q_{N}, p_{N})$ be a path between these two configurations and suppose it stays at a distance at most $|Q|$ from the picture. For every $i \in \llbracket 0, N\rrbracket$ we have $p_{i} = (p_{i}^{(1)}, \dots, p_{i}^{(d)}) \in \mathbb{Z}^d$. We can construct a sequence that simulates the path:
\begin{itemize}
\item Start with $(q'_{0}, p'_{0})$ where
\item $p'_{0} = (p_{0}^{\prime(1)}, \dots, p_{0}^{\prime(d)})$ where $\forall i \in \llbracket 1, d\rrbracket, p_{0}^{\prime(i)} = (\operatorname{cases}(1 \quad \mathrm{if} p_{0}^{(i)} < 1, p_{0}^{(i)} \quad \mathrm{if} 1 \le p_{0}^{(1)} \le d_{i}, d_{i} \quad \mathrm{if} d_{i} < p_{0}^{(i)}))$
\item $q'_{0} = (q_{0}, p_{0} - p'_{0}) \in Q'_{\mathrm{in}}$
\end{itemize}

\item For all $i \in \llbracket 0, N-1\rrbracket$, we can construct $(q'_{2i+1}, p'_{2i+1})$ and $(q'_{2i+2}, p'_{2i+2})$ from $(q'_{2i}, p'_{2i})$ with $q'_{2i} = (q'', p'') \in Q'_{\mathrm{in}}$:
\begin{itemize}
\item $q'_{2i+1} = (q_{i+1}, p_{i+1}-p_{i}, p'') \in Q'_{\mathrm{check}}$
\item $p'_{2i+1} = p'_{2i}$
\item $(q'_{2i+2}, m) = \delta'(q'_{2i+1}, b_{p'_{2i+1}}) \in Q'_{\mathrm{in}}$
\item $p'_{2i+2} = p'_{2i} + m$
\end{itemize}

By construction, we have $f(q'_{0}, p'_{0}) = (q_{0}, p_{0})$ and $\forall j \in \llbracket 0, 2N\rrbracket, p'_{j} \in \operatorname{edom}(b)$ and also $f(q'_{2N}, p'_{2N}) = (q_{N}, p_{N})$.

Reciprocally, if we have a sequence of configuration of the NFA, by construction, it simulates one path between two configurations of the FNFA that stays at a distance at most $C$ from the picture.
\end{proof}

\begin{theorem}
Let $A = (\Sigma, Q, M, I, F, \delta)$ be a picture walking FNFA on $\mathbb{Z}^d$ and $b \in \Sigma_*^*$ be a picture. There exists a picture walking NFA $A' = (\Sigma, Q', M, I', F', \delta')$ such that $L(A) = L(A')$.
\end{theorem}

\begin{proof}
Let $A = (\Sigma, Q, M, I, F, \delta)$ be a picture walking FNFA on $\mathbb{Z}^d$ and $b \in \Sigma_*^*$ be a picture. We want to construct a NFA $A' = (\Sigma, Q', M, I', F', \delta')$ that recognizes the same language as $A$. For the sake of simplicity, with Proposition~\ref{prop:SimulateClose} we will assume that $A'$ can freely move outside the picture as long as it stays at a distance at most $|Q|$ from the picture. Moreover, we will also assume that $A$ need to come back to the picture in order to accept a word.

Now, consider $b \in L(A)$, therefore there exists two configurations $(q_{s}, p_{s})$ and $(q_{t}, p_{t})$ such that:
\begin{itemize}
\item $q_{s} \in I$ and $p_{s} = (1, \dots, 1)$
\item $q_{t} \in F$ and $p_{t} \in \operatorname{dom}(b)$
\item There exists a path between $(q_{s}, p_{s})$ and $(q_{t}, p_{t})$
\end{itemize}
Consider $w = (q_{0}, p_{0}), \dots, (q_{N}, p_{N})$ the minimal path between $(q_{s}, p_{s})$ and $(q_{t}, p_{t})$. This path can be divided into $r$ subpaths $w_{1}, \dots, w_{r}$ for some $r \in \mathbb{N}^*$ such that:
\begin{itemize}
\item $w = w_{1} \dots w_{r}$
\item $\forall \in \llbracket 1, r\rrbracket, i \mathrm\{odd\} \Rightarrow w_{i}$ stays at a distance at most $|Q|$ from the picture.
\item $\forall \in \llbracket 1, r\rrbracket, i \mathrm\{even\} \Rightarrow w_{i}$ stays at a distance at least $|Q|$ from the picture.
\end{itemize}
By Proposition 1, we can already simulate $w_{i}$ when $i$ is odd.
For $i$ even, we know that $w_{i}$ are minimal paths that stay at a distance at least $|Q|$ from the picture, therefore by Lemma 2, there exists $M \in \mathbb{N}$ such that each of their good crossed half-space sequences are length at most $M$. The only that remains to do is to show that these paths can be simulated by a NFA.

We will show this result by induction on the length of crossed-half spaces sequence. For all $L \in \mathbb{N}^*$, for all $q, q' \in Q$ and for all $H$ half-space, let $P_{L, q, q', H}$ be the following property $P_{L, q, q', H}$:

\vspace{0.2cm}
\noindent
\fbox{\parbox{\textwidth}{There exists an NFA $A'$ such that for every path $(q_{0}, p_{0}), \dots, (q_{N}, p_{N})$ with $q_{0} = q$, $q_{N} = q'$ and $p_{0}, p_{N}$ at a distance $|Q|$ from the picture, and $(q_{1}, p_{1}), \dots (q_{N-1}, p_{N-1})$ a path that starts in the half-space $H$, staying at a distance at least $|Q|$ from the picture and having a good crossed-half-space sequence of length $L$, the automaton $A'$ can go from $p_{0}$ to $p_{N}$ (inside the picture); but such a transition is not possible otherwise.}}
\vspace{0.2cm}

\textbf{Base Case:} Let $L = 1$ and $q, q' \in Q$ be two states and $H$ a half-space. There is a push-down automaton $B$ that can compute all sequence of moves that gives path of the property:
\begin{itemize}
\item $B$ starts with an empty stack and in state $q$
\item $B$ simulates $A$ if it was reading only $\#$
\item Each time $A$ goes one step deeper in $H$, $B$ add a character on its stack
\item Each time $A$ goes one step back in $H$, $B$ delete a character of its stack
\item $B$ accepts only if the stack is empty and if we are in state $q'$
\end{itemize}

This PDA recognizes exactly paths of the property. By the Parikh's theorem the set of Parikh vectors of this language is a semi-linear set. Therefore by properties of semi-linear sets the set of position we could have landed is also exactly a semi-linear set i.e. it has the shape $\bigcup_{j=1}^N \{v_{j} + \sum_{i=1}^{n_j} a_{i, j} v_{i, j} | a_{1, j}, \dots, a_{n_{j}, j} \in \mathbb{N}\}$ for some vectors $v_{1}, \dots, v_{N}, v_{1,1}, \dots, v_{n_{j}, N} \in \mathbb{Z}^d$ and $N \in \mathbb{N}$. Now we can create a NFA $A'$ that can go in each of these positions.

From the Steinitz Lemma \cite{Steinitz,Ambrus_2026}, we know that is possible to go in each of these positions while staying at a bounded distance from the segment defined by the starting and ending position. Thus, to simulate it we just have to allow the NFA to go at a bounded distance outside the picture.

\textbf{Base Case:} Let $L > 1$ such that $\forall q, q', H$ we have $P_{L-1, q, q', H}$. Let $q, q' \in Q$ be two states and $H$ a half-space. For all states $q'' \in Q$, there is NFA $A'_{q''}$ that can go in every positions we could have landed after $L-1$ crossed half space. For every $q''$ and every $(L-1)^{\mathrm{th}}$ crossed half-space, if the automaton tries to land at a distance strictly more than $|Q|$ from the picture, then the automaton would have encountered an edge of the picture. We call $H'$ the half-space where the automaton wanted to go and it is the $L^{\mathrm{th}}$ half-space of the sequence.

There is a push-down automaton $B_{A'_{q''}, H'}$ that can accept the paths that end correctly:
\begin{itemize}
\item $B_{A'_{q''}, H'}$ starts with an empty stack
\item $B_{A'_{q''}, H'}$ simulates $A'_{q''}$ in $H'$
\item $B_{A'_{q''}, H'}$ then simulates $A$ starting from state $q''$ as if it was only reading $\#$
\item Each time the simulation goes one step deeper in $H'$, $B_{A'_{q''}, H'}$ adds a character on its stack
\item Each time the simulation goes one step back in $H'$, $B_{A'_{q''}, H'}$ deletes a character of its stack
\item $B_{A'_{q''}, H'}$ accepts only if the stack is empty and if we are in state $q'$
\end{itemize}

The union recognizes exactly ends of paths of the property. By the Parikh's theorem the set of Parikh vectors of $L(B_{A'_{q''}, H'})$ is a semi-linear set.

Therefore by properties of semi-linear sets the set of positions we could have landed in is also exactly a semi-linear set i.e. it has the shape $\bigcup_{j=1}^N \{v_{j} + \sum_{i=1}^{n_j} a_{i, j} v_{i, j} | a_{1, j}, \dots, a_{n_{j}, j} \in \mathbb{N}\}$ for some vectors $v_{1}, \dots, v_{N}, v_{1,1}, \dots, v_{n_{j}, N} \in \mathbb{Z}^d$ and $N \in \mathbb{N}$.

Now we can create a NFA $A''$ that can go in each of these positions. As NFA are closed under union, we can finally create an automaton such that $P_{L_, q, q', H'}$ is true.

Let $A''$ be the NFA which is the union of each NFA of each $P_{L, q, q', H}$ for all $L \le M, q \in Q, q' \in Q$ and $H$ half-spaces. This automaton can simulate minimal paths that stay at a distance at least $|Q|$ from the picture.

Therefore, it is possible to construct a NFA $A'$ prove our theorem.
\end{proof}

\begin{corollary}
FNFA = NFA in all dimensions.
\end{corollary}

\bibliographystyle{plainnat}
\bibliography{references}
\end{document}